\documentclass[preprint,amsmath,amssymb]{revtex4-1}
 \usepackage{amsthm}
\newtheorem{thm}{Theorem}[section] 
\newtheorem{lemma}[thm]{Lemma}

\newcommand{\bra}[1]{\langle#1|}

\newcommand{\ket}[1]{|#1\rangle} 
\newcommand{\braket}[2]{\langle#1|#2\rangle} 
\newcommand{\ketbra}[2]{|#1\rangle \langle#2|} 
\newcommand{\pbrax}[2]{  \bra{\Phi} \hat{P}_{#1} \hat{P}_{#2} \ket{\Phi}  }
\newcommand{ \pbrac}[1]{  \bra{\Phi} \hat{P}_{#1}  \ket{\Phi} }

\begin{document}

 \title{Entanglement enhancement of a noisy classical communication channel}
\date{September 6, 2011}
\author{H. Thomas Williams}
%\affiliation{Department of Physics and Engineering, Washington and Lee University, Lexington, VA  24450 USA}
\email{williamsh@wlu.edu}
\author{Paul Bourdon}
%\affiliation{Department of Mathematics, Washington and Lee University, Lexington, VA  24450 USA}
\email{bourdonp@wlu.edu}
\begin{abstract}   We present a formal quantum mechanical analysis of the communication protocol of Prevedel {\it et al.}\ [Phys. Rev. Lett. \textbf{106}, 110505 (2011)], in which entanglement shared by sender and receiver is used to enhance, beyond that achievable via the optimal classical strategy,  the probability of successful transmission of a bit through a particular noisy classical channel ${\cal N}$. We provide a full analysis of this protocol when the shared entanglement resides in a two-qubit system.  Our analysis shows the measurement choices specified by the protocol yield the maximum possible enhancement of the probability of successful communication of the bit with one use of the channel ${\cal N}$.  We determine that shared entanglement residing in a two-qudit system with $d > 2$ cannot provide enhancement beyond that produced by entangled qubits. Finally, we show how the protocol should be extended when probability parameters for the channel ${\cal N}$ are allowed to vary.  
 \end{abstract} 
\maketitle

\section{Introduction}
The requirements of causality prohibit quantum entanglement alone  from serving as a medium of communication; however, it can quantitatively enhance
communication in a variety of settings.  For example, entanglement exactly doubles the classical capacity of a noiseless quantum channel via
superdense coding \cite{bennett2}, and can provide arbitrarily large capacity-enhancement factors for noisy quantum  channels \cite{bennett3}.

Recent attention has focused as well on the issue of  enhancing the reliability of classical communication over noisy classical channels via quantum entanglement.  While shared entanglement between sender and receiver cannot increase the capacity of noiseless classical channels \cite{bennett1}, nonetheless some noisy classical channels become more effective when supplemented by enhancement.  Cubitt {\it et al.}\ \cite{cubitt} considered the effect of shared entanglement on zero-error classical communication in the case of a single use of a classical channel.  Such classical channels take multiple inputs, and one can identify a maximum set of inputs that lead to distinct distinguishable outputs:  the size of this set is a measure of the zero-error, one-shot capacity of the channel.  In particular, this theoretical work included consideration of a channel with 24 inputs and 18 outputs.  The classical channel alone allows a maximum of 5 messages to be sent error-free in each distinct use of the channel, but 6 or more error-free messages can be sent for each use of the channel when supplemented by shared quantum entanglement.  A much simpler noisy classical channel  ${\cal N}$, one with 4 inputs and 6 distinct, randomly appearing outputs, has been examined by Prevedel \textit{et al.}\ \cite{prevedel}.  They  presented a protocol providing a significant enhancement of the probability of successful communication of one bit in a single use of the  channel ${\cal N}$ through clever exploitation of shared entanglement.  Furthermore, they provided experimental verification of the enhancement.

Herein, we examine in some detail the quantum mechanics of the communication protocol introduced in reference \cite{prevedel}. We begin by establishing a formal expression for the success probability of protocols having the structure of that from \cite{prevedel}.  This expression is based upon (i) the quantum state of the entangled pair shared by sender Alice and receiver Bob, (ii) the quantum states representing the measurements made by Alice on her particle from the entangled pair (based on what message is to be sent), and (iii) the quantum states representing the measurements made by Bob on his particle (based on the output of channel ${\cal N}$ and designed to increase the chance that he will correctly determine Alice's message).   This expression is then used  to examine fully the case in which  Alice and Bob share an entangled two-qubit system, exactly the situation tested by Prevedel {\it et al.}  We produce a rigorous proof that the protocol presented in \cite{prevedel} is the best possible in the sense that the measurement strategy it specifies  provides the maximum enhancement of the probability that Alice will successfully communicate to Bob one bit through one use of the channel ${\cal N}$.  Following that we generalize the protocol by considering entangled qudits of arbitrary dimensionality, and determine that (perhaps counterintuitively) entangled qudits with dimension greater than two do not enhance communication beyond that produced by entangled qubits.  Finally, we show how the protocol should be extended when probability parameters for the channel ${\cal N}$ are allowed to vary. 

\section{Communication protocol \label{comprot}}
Alice signals Bob via a noisy classical channel $\cal N$.  $\cal N$ takes as input two bits $(q_1,q_2)$ and delivers as output a trit and a bit of one of the forms $(1,q_1)$, $(2,q_2)$, or $(P,q_1 \oplus q_2)$, where the third, parity option returns in the bit place the sum of the two sent bits, modulo 2. For the moment we assume the three output forms specified by the the trit appear  with equal probabilities.  The equal-probability assumption will be relaxed in section \ref{genprob}.  Use of this channel alone to communicate a single bit $q$ can give Bob a $5/6$ probability of properly decoding the message $q$.  One straightforward way to achieve this is for Alice to send $(q,q)$: Bob will receive one of three equally probable channel outputs,  $(1,q)$, $(2,q)$, or $(P,0)$.
  In the first two cases he knows the message, and although the parity output carries no information, he has a 50\% chance of guessing the value of $q$ (with the assumption that messages $q=0$ and $q=1$ are sent with equal probability).

If Alice and Bob share the quantum resource of an entangled pair of particles  in state $\ket{\Phi}$, it was shown by Prevedel {\it et al.}\ in the case of shared qubits, that the $5/6$ classical-decoding probability can be improved upon.  The protocol to accomplish this is as follows: 
\begin{itemize} 
   \item If Alice wishes to communicate the message $q=0$, she measures her half the entangled state along the axis $\ket{\psi}$, and if her measurement returns an outcome in the direction she measures (``true") she sets a bit $\alpha$ to value $0$, otherwise $\alpha = 1$; if she wishes to send $q=1$, she measures along $\ket{\psi'}$, setting the value $\alpha=0$ for ``true" and $\alpha=1$ otherwise; 
   \item Alice inputs $(q,\alpha)$ into the channel $\cal N$;
   \item If Bob receives $(1,q)$ he knows the message $q$ and does nothing more; 
   \item If Bob receives $(2,\alpha)$ he measures his half the entangled state along $\ket{\eta'}$, setting variable $\beta = 0$ if the measurement yields ``true," setting $\beta=1$ otherwise: if $\beta = \alpha$, Bob assumes Alice has measured along $\ket{\psi}$ and the message is $q=0$, otherwise he assumes the message is $q = 1$; 
   \item If Bob receives $(P,q \oplus \alpha)$ he measures his half the entangled state along $\ket{\eta}$ and sets $\beta= 0$ if the measurement gives ``true,''  setting $\beta=1$ otherwise:  Bob assumes in this case that $\alpha = \beta$; if this is indeed the case, he can extract the message $q$ from the received value $q \oplus \alpha$ (by adding $\beta$).
\end{itemize}
   \subsection{Calculation of probability of successful transmission of $q$}
 When Bob receives $(1,q)$ he immediately knows Alice's message, contributing a factor $1/3$ to his success rate.  For the other channel outputs, his overall success rate in deciphering the message depends on Alice's measurement directions and his own.  In order to calculate quantum measurement probabilities, we utilize projection operators, e.g. $\hat{P}_{\psi} \equiv \ketbra{\psi}{\psi}$.  The probability of measuring a state $\ket{\phi}$ along $\ket{\psi}$ and finding the particle so aligned is $\bra{\phi} \hat{P}_{\psi} \ket{\phi}$;  the probability that the measurement yields ``false" is $\bra{\phi} (I - \hat{P}_{\psi}) \ket{\phi}$.  Two-particle product states in what follows will be written as linear combinations of kets of the form $\ket{\psi,\eta}$ with the first entry being Alice's particle state, and the second being Bob's.  

When Bob receives $(2,\alpha)$ he successfully decodes the message $q$ when $q=0$ and $\alpha = \beta$ or when $q=1$ and $\alpha \neq \beta$.  A detailed success analysis appears below, with (a) and (b) pertaining to the $q=0$ message and (c) \ and (d) pertaining to the message $q=1$: 
\begin{itemize}
\item[(a)] Alice measures along $\ket{\psi}$ getting $\alpha = 0$ and Bob (measuring along $(\ket{\eta'})$) gets $\beta=0$, with probability
\[ \pbrax{\psi}{\eta'} ; \] 
\item[(b)]  Alice measures along $\ket{\psi}$ getting $\alpha = 1$ and Bob gets $\beta=1$, with probability
\[  \bra{\Phi}(I_A - \hat{P}_{\psi} )((I_B - \hat{P}_{\eta'})   \ket{\Phi}  = 1 + \pbrax{\psi}{\eta'} - \pbrac{\psi} - \pbrac{\eta'};\ \] 
\item[(c)] Alice measures along $\ket{\psi'}$ getting $\alpha = 0$ and Bob gets $\beta=1$, with probability
\[  \bra{\Phi}\hat{P}_{\psi'}(I_B - \hat{P}_{\eta'} )  \ket{\Phi}  =  -  \pbrax{\psi'}{\eta'} + \pbrac{\psi'}  ;\] 
 \item[(d)] Alice measures along $\ket{\psi'}$ getting $\alpha = 1$ and Bob gets $\beta=0$, with probability  
\[  \bra{\Phi}(I_A - \hat{P}_{\psi'} )\hat{P}_{\eta'}  \ket{\Phi}  =  - \pbrax{\psi'}{\eta'} + \pbrac{\eta'}  .\] 
\end{itemize}
Operators $I_A$ and $I_B$ are identity operators in the space of Alice's and Bob's particles, respectively.  We omit these operators when context makes an expression's meaning obvious, e.g. $\bra{\Phi}I_A  \hat{P}_{\eta'}   \ket{\Phi}$ is written as $\bra{\Phi} \hat{P}_{\eta'} \ket{\Phi}$.
Because the output of channel $\cal N$ is $(2,\alpha)$ with probability $1/3$, while there's a $1/2$ chance that $q=0$ and a $1/2$ chance that $q=1$, the probability that Bob obtains the message $q$ when the channel output is $(2,\alpha)$ is given by one-sixth the sum of the probabilities given in (a) through (d) above, thus
\[ \frac16 \left( 1 + 2 \pbrax{\psi}{\eta'} -2\pbrax{\psi'}{\eta'} + \pbrac{\psi'} - \pbrac{\psi} \right) . \]

When Bob receives $(P,q \oplus\alpha)$ he successfully deduces $q$ precisely when $\alpha = \beta$ independent of the value of $q$.  Here are the details, with  (e) and (f) describing success when $q=0$  and (g) and (h) describing success when $q=1$:  
\begin{itemize}
\item[(e)]  Alice measures along $\ket{\psi}$ and $\alpha = \beta = 0$, with probability
\[   \pbrax{\psi}{\eta}; \]
\item[(f)] Alice measures along $\ket{\psi}$ and $\alpha = \beta = 1$, with probability
\[  \bra{\Phi}(I - \hat{P}_{\psi} )(I - \hat{P}_{\eta} )  \ket{\Phi}  = 1 + \pbrax{\psi}{\eta} - \pbrac{\psi} - \pbrac{\eta} ;\] 
\item[(g)] Alice measures along $\ket{\psi'}$ and $\alpha = \beta = 0$, with probability
\[     \pbrax{\psi'}{\eta}   ;\] 
\item[(h)] Alice measures along $\ket{\psi'}$ and $\alpha = \beta = 1$, with probability
\[ \bra{\Phi}(I - \hat{P}_{\psi'} )(I - \hat{P}_{\eta})\ket{\Phi}  = 1 + \pbrax{\psi'}{\eta} - \pbrac{\psi'}  -  \pbrac{\eta}  .\] 
\end{itemize}
Because this final case (reception of $(P,q \oplus \alpha)$) occurs with probability $1/3$ and, just as before, the probabilities that $q=1$ and $q=0$ are each $1/2$, we find the probability of success in this case to be one-sixth the sum of the probabilities given in (e) through (h) above, thus
\[ \frac16 \left( 2 + 2 \pbrax{\psi}{\eta} + 2 \pbrax{\psi'}{\eta} - 2 \pbrac{\eta} - \pbrac{\psi} - \pbrac{\psi'}  \right) . \] 

The overall success rate of the protocol is the sum of the success probabilities for the three channel outputs that Bob can receive:
\begin{equation}
 S = \frac56 +  F(\ket{\Phi}, \ket{\psi}, \ket{\psi'}, \ket{\eta}, \ket{\eta'}) \label{success}
\end{equation}
with
\begin{eqnarray}
   F(\ket{\Phi}, \ket{\psi}, \ket{\psi'}, \ket{\eta}, \ket{\eta'}) = \frac13 \left(\right. \pbrax{\psi}{\eta} &+& \pbrax{\psi}{\eta'} + \pbrax{\psi'}{\eta} - \pbrax{\psi'}{\eta'}  \nonumber \\
 &-&  \pbrac{\psi} - \pbrac{\eta}  \left.\right). \label{enhancement}
\end{eqnarray}
The function $F$ represents the departure of the success rate for the entanglement-enhanced communication protocol from that of the best rate achievable using the classical channel alone($5/6$).  It can be positive (as exhibited by the Prevedel result) as well as negative (in that one can pick measurement directions that cause the assisted protocol to perform worse than the simple classical channel ideal).  The remainder of this paper focuses on maximization of $F$, our objective function.

\section{Shared qubits \label{qubits}}
  Consider the maximization of the objective function when the shared quantum resource is two qubits.  Without loss of generality, we can assumed the shared state to have the form of the Schmidt decomposition \cite[p.\ 109]{neilsen}
\[ \ket{\Phi} = \sqrt{L_0}\, \ket{0,0} +\sqrt{1-L_0}\, \ket{1,1} , \]
with $L_0$ real and non-negative.
Likewise, with full generality we can define
\[ \ket{\psi} = a \ket{0} + \exp(i\alpha) \sqrt{1 - a^2} \, \ket{1} \;\;,\;\; \ket{\psi'} = a' \ket{0} + \exp(i\alpha') \sqrt{1 - {a'}^2}\,  \ket{1}\]
 as kets in the space of Alice's particle, where $a$, $a'$ are real and non-negative, and $\alpha$, $\alpha'$ are real phase angles, and
\[ \ket{\eta} = b \ket{0} + \exp(i\beta) \sqrt{1 - b^2}\,  \ket{1} \;\;,\;\; \ket{\eta'} = b' \ket{0} + \exp(i\beta') \sqrt{1 - {b'}^2}\,  \ket{1} \]
as kets in the space of Bob's particle, where $b$, $b'$ are real and non-negative, and $\beta$, $\beta'$ are real phase angles.
Thus, for example
\[ \pbrax{\psi}{\eta} = \left| a b \sqrt{L_0} + \exp(-i(\alpha + \beta)) \sqrt{1 - a^2}\sqrt{1 - b^2} \sqrt{1 - L_0} \right|^2 \]
and 
\[ \pbrac{\psi} = a^2  L_0 + (1-a^2)(1-L_0) . \]
Using these as templates along with Eq.~(\ref{enhancement}), we obtain in the special case $L_0=1$, 
\[ F_{L_0=1} = \frac13 \left( -a^2 - b^2 -{a'}^2{b'}^2+a^2b^2+a^2{b'}^2+{a'}^2 b^2 \right) , \]
and in the case $L_0 = 0$, the identical result:  $F_{L_0=0} = F_{L_0 = 1}$.

Expressing the general form of (\ref{enhancement}), showing explicitly all dependence on $L_0$, we have
\begin{eqnarray}
 F(\Phi, \psi, \psi', \eta, \eta') &=& L_0 F_{L_0=1} +(1-L_0)F_{L_0=0}+ \sqrt{L_0(1-L_0)} G(a,a',b,b',\alpha, \alpha',\beta,\beta') \nonumber \\
  &=& F_{L_0=0} + \sqrt{L_0(1-L_0)} G(a,a',b,b',\alpha, \alpha',\beta,\beta'). \nonumber
\end{eqnarray}
The term containing $G$ comes from the cross terms of the squares in the first four terms on the RHS of Eq.~(\ref{enhancement}), and it includes all the $L_0$ dependence in this result.  When $L_0=0$, the shared resource $\ket{\Phi}$ is a product state, without entanglement.  Since in this case there can be no enhancement of the result obtained using the classical communication channel $\cal N$ alone, we must have $F_{L_0=0} <=0$.  This implies that when entanglement enhancement does take place, the term containing $G$ must be positive. This term is maximized when $L_0 = 1/2$.  Entanglement enhancement is optimized, therefore, for full entanglement
\begin{equation} \ket{\Phi} = \frac1{\sqrt{2}} ( \ket{0,0} + \ket{1,1} ) . \label{fullent} \end{equation}

There remains some unitary freedom that can further simplify our maximization process.  Since the shared state (\ref{fullent}) is one of full entanglement, we can orient Alice's coordinate system using an arbitrary unitary transformation $U_A$, as long as we simultaneously transform Bob's states using $U^*_B$ (see Sec.\ A of the Appendix).  This unitary transformation $U_A \otimes U^*_B$ will allow us to pick Alice's state $\ket{\psi'} = \ket{0}$, leaving the form of the shared state, Eq.~(\ref{fullent}), invariant.

This choice for $\ket{\psi'}$ simplifies the objective function (assuming full entanglement) to the form
\begin{eqnarray}
   F = \frac12 &\left( \right.& \left| a b  + \exp(-i(\alpha + \beta)) \sqrt{1 - a^2}\sqrt{1 - b^2}  \right|^2 + \left| a b'  + \exp(-i(\alpha + \beta')) \sqrt{1 - a^2}\sqrt{1 - b'^2}  \right|^2  \nonumber \\
    &\;& \; + {b}^2  - {b'}^2 - 2 \left. \right) . \nonumber 
\end{eqnarray}
Each of the first two terms on the RHS of this expression are non-negative, and are maximized when the complex phase factors within the square are $+1$.  A simple choice that accomplishes this is  $\alpha = \beta = \beta' =0$; thus all components of the quantum states $\ket{\psi}$, $\ket{\psi'}$, $\ket{\eta}$, and $\ket{\eta'}$ are real and non-negative. With this choice
 \begin{equation}
   F = \frac12 \left( ( a b  + \sqrt{1 - a^2}\sqrt{1 - b^2}  )^2 + ( a b'  +  \sqrt{1 - a^2}\sqrt{1 - b'^2}  )^2  + {b}^2  - {b'}^2 - 2  \right) . \label{qbfin} 
\end{equation}
Note that the objective function has domain $0 \leq a \leq 1, \; 0 \leq b \leq 1, \; 0 \leq b' \leq 1$;  its maximum value is
\[ F_{max} = (4 \cos^2(\pi/8)-3)/6, \]
which occurs at the only critical point of $F$  lying in the interior of the domain: $a = \cos(\pi/4)$, $b = \cos(\pi/8)$, and $b' = \cos(3\pi/8)$.  

These values of $a$, $b$, and $b'$ as well as the associated maximum enhancement  in the success rate that they yield for the channel ${\cal N}$ (namely $F_{max}$) correspond precisely to protocol presented in Prevedel {\it et al.}.  This provides a proof that the enhancement exhibited by Prevedel {\it et.al.} is optimal.   

\section{Shared qudits \label{qudits}}
We have explored the efficiency of the same classical channel $\cal N$ when Alice and Bob share two entangled qudits.  Intuitively one might assume that this apparently richer resource would allow a yet greater enhancement in Bob's ability to decode Alice's message, but this can be shown not to be the case, as follows.

Assume Alice and Bob share an entangled state of two qudits.  They will utilize this resource just as described in the previous section, Alice measuring along $\ket{\psi}$ and $\ket{\psi'}$ and Bob along $\ket{\eta}$ and $\ket{\eta'}$, with these four states also $d$-dimensional qudits. Create an orthonormal basis of d-dimensional quantum states $\ket{p_i},\;i=0,1,\ldots,d-1$, in Alice's space, using the Gram-Schmit process, with $\ket{\psi'} = \ket{p_0}$ and such that $\ket{\psi}$ can be written as a linear combination of $\ket{p_0}$ and $\ket{p_1}$.  Similarly, create an orthonormal basis $\ket{n_i}$ in Bob's space,  with $\ket{\eta} = \ket{n_0}$ and such that $\ket{\eta'}$ can be written as a linear combination of $\ket{n_0}$ and $\ket{n_1}$.  One can write a general expression for a shared pair of qudits, normalized, in terms of these bases:
\begin{equation}\label{PhiQD}
 \ket{\Phi} = \sum_{i,j=0}^{d-1} m_{ij}\ket{p_i}\ket{n_j}. 
\end{equation}
By truncating the sum we write a related, possibly unnormalized shared state
\begin{equation}\label{PhiTQD}
 \ket{\tilde{\Phi}} = \sum_{i,j=0}^{1} m_{ij}\ket{p_i}\ket{n_j}. 
\end{equation}

Due to the alignments of Alice's and Bob's measurement directions, $\pbrax{\psi}{\eta} = \bra{\tilde{\Phi}}\hat{P}_{\psi}\hat{P}_{\eta}\ket{\tilde{\Phi}}$, $\pbrax{\psi}{\eta'} = \bra{\tilde{\Phi}}\hat{P}_{\psi}\hat{P}_{\eta'}\ket{\tilde{\Phi}}$, $\pbrax{\psi'}{\eta} = \bra{\tilde{\Phi}}\hat{P}_{\psi'}\hat{P}_{\eta}\ket{\tilde{\Phi}}$, and $\pbrax{\psi'}{\eta'} = \bra{\tilde{\Phi}}\hat{P}_{\psi'}\hat{P}_{\eta'}\ket{\tilde{\Phi}}$.  It is straightforward to establish, as well, that 
\[  \bra{{\Phi}}\hat{P}_{\psi}\ket{{\Phi}} \ge \bra{\tilde{\Phi}}\hat{P}_{\psi}\ket{\tilde{\Phi}} \quad \text{and}  \bra{{\Phi}}\hat{P}_{\eta}\ket{{\Phi}} \ge \bra{\tilde{\Phi}}\hat{P}_{\eta}\ket{\tilde{\Phi}} . \]
Inserting these relationships in the general form of the objective function, Eq.~(\ref{enhancement}), it is easy to see that
\[ F(\ket{\Phi}, \ket{\psi_0},\ket{\psi_1}, \ket{\eta_0}, \ket{\eta_1}) \le F(\ket{\tilde{\Phi}}, \ket{\psi_0},\ket{\psi_1}, \ket{\eta_0}, \ket{\eta_1}).  \]
But for the issue of normalization of $\ket{\tilde{\Phi}}$, the right-hand side of this expression is identical to the result obtained using shared qubits, thus
\begin{eqnarray}
 \max(F(\ket{\Phi}, \ket{\psi_0},\ket{\psi_1}, \ket{\eta_0}, \ket{\eta_1})) &\le& \max(F(\ket{\tilde{\Phi}}, \ket{\psi_0},\ket{\psi_1}, \ket{\eta_0}, \ket{\eta_1}) \nonumber \\ &\le& \braket{\tilde{\Phi}}{\tilde{\Phi}} \max(F(\ket{\Phi_{qubit}}, \ket{\psi_0},\ket{\psi_1}, \ket{\eta_0}, \ket{\eta_1}) . \label{qudit}
\end{eqnarray}

Since $\braket{\tilde{\Phi}}{\tilde{\Phi}} \le 1$, Eq.~(\ref{qudit}) establishes the fact that the optimal entanglement-assisted communication of a single classical bit, using classical communication channel $\cal{N}$ and the Prevedel protocol, can be achieved using shared qubits.

It is possible to prove that if the system that Alice and Bob share comprises two fully entangled qudits, then the maximum enhancement value achievable is 
\[ \frac2d F_{max} , \]
where $F_{max} = (4 \cos^2(\pi/8)-3)/6$ is the maximum enhancement possible for qubits (see Sec.\ B of the Appendix.) Thus, for $d \geq 3$, fully entangled qudits do not allow for as high a success rate as fully entangled qubits.

\section{Generalized probabilities \label{genprob}}
A straightforward generalization of the classical channel $\cal N$ of Prevedel {\it et al.}\ comes from dropping the restriction that Bob receives the three possible channel outputs with equal probability.  Define the channel $\cal N'$ as the classical channel for which when Alice sends input $(q_1, q_2)$, Bob will receive output $(1,q_1)$ with probability $c_1$, output $(2,q_2)$ with probability $c_2$, and output $(P,q_1 \oplus  q_2)$ with probability $c_3$:  $c_1+c_2+c_3 = 1$. Even using this channel alone, the maximum probability of Alice's communicating a single bit ($q$) to Bob depends upon which two bits she communicates:  if she inputs $(q,q)$, the maximum probability is $(2c_1+2c_2+c_3)/2=1-c_3/2$; if she inputs $(q,0)$ it is $(2c_1+c_2+2c_3)/2= 1-c_2/2$; and if she inputs $(0,q)$ it is $(c_1+2c_2+2c_3)/2=1-c_1/2$.  Obviously the channel input should be chosen based upon which of the output probabilities is the smallest.  

When considering the degree of enhancement that can be provided by the sharing of entangled particles by Alice and Bob, one likewise needs to chose measurement strategies based on the relative probabilities of the three channel outputs.  We first examine the case in which $c_2$ is less than or equal to both $c_1$ and $c_3$.  The highest probability of communication of a single bit using the channel $\cal N'$ alone is $1-c_2/2$.  As in section \ref{qubits}, we assume Alice and Bob share entangled qubits;  Alice places $(q,\alpha)$ into $\cal N'$ (where $\alpha$ is the result of the measurement on her qubit, along $\ket{\psi}$ to send $q=0$, and along $\ket{\psi'}$ to send $q=1$); and Bob measures along different axes depending upon the channel output he receives (along $\ket{\eta'}$ when he receives $(2,\alpha)$, and along $\ket{\eta}$ when he receives $(P,q \oplus \alpha)$.)

Using results from section \ref{comprot} we establish Bob's success rate in this more general case as
\begin{eqnarray*}
   S' = 1 - \frac{c_2}2 + &c_2& \left(  \pbrax{\psi}{\eta'} - \pbrax{\psi'}{\eta'} - \pbrac{\psi}/2 + \pbrac{\psi'}/2 \right)  \\
   + &c_3& \left(  \pbrax{\psi}{\eta} + \pbrax{\psi'}{\eta} - \pbrac{\eta} - \pbrac{\psi}/2 - \pbrac{\psi'}/2 \right) .
\end{eqnarray*}
Assuming the shared states are qubits, and using the definitions of $L_0, a, a', b, b'$ and the phase angles $\alpha, \alpha', \beta, \beta'$ introduced in section \ref{qubits}, we follow the process of that section to first show
\[ S'_{L_0 = 1} = 1 - \frac{c_2}2 + c_2(a^2-{a'}^2) ({b'}^2-\frac12 ) + c_3\left( (a^2+{a'}^2) ({b}^2-\frac12 ) - b^2  \right) \] 
and
\[ S'_{L_0 = 1} = S'_{L_0 = 0} . \]
The logic for the general case exactly parallels that of section \ref{qubits}, leading to the result
\begin{eqnarray}
  S'(\Phi, \psi, \psi', \eta, \eta') &=& L_0 S'_{L_0=1} +(1-L_0)S'_{L_0=0}+ \sqrt{L_0(1-L_0)} G'(a,a',b,b',\alpha, \alpha',\beta,\beta') \nonumber \\
  &=& S_{L_0=0} + \sqrt{L_0(1-L_0)} G'(a,a',b,b',\alpha, \alpha',\beta,\beta'). \nonumber
\end{eqnarray}
The term including the factor $G'$ represents the enhancement of the success rate over that of the unentangled case $L_0 = 0$, and will be positive at its maximum value.  As can be seen by the $L_0$-dependence of the expression, enhancement is optimized when $L_0 = 1/2$, full entanglement.
 
As we have shown previously,  we can assume $\ket{\psi'} = \ket{0}$, allowing us to write the expression for successful deciphering as
\begin{eqnarray*}
   S' = 1 - \frac{c_2}{2} + &\frac{c_2}{2}& \left(  | ab'+\exp(-i(\alpha+\beta'))\sqrt{1-a^2}\sqrt{1-{b'}^2} |^2 - {b'}^2  \right)  \\
   + &\frac{c_3}2& \left(  | ab+\exp(-i(\alpha+\beta))\sqrt{1-a^2}\sqrt{1-{b}^2} |^2 + b^2 -2 \right) .
\end{eqnarray*}
The two terms in this expression with complex phase factors are both positive, and are both maximized by choosing the phase angles to be zero, which we do henceforth.  What remains in the process of maximizing $S'$ is a straightforward extremum problem in variables $a$, $b$, and $b'$, which (after some tedious algebra) yields
\begin{equation} {S'}_{max} = 1 + \frac{1}2 (\sqrt{c_2^2+c_3^2} - c_2 - c_3)  \label{sprime} \end{equation}
resulting from choices $a = \cos(\theta)$, $b = \cos(\theta/2)$, and $b' = \cos(\theta/2+\pi/4)$ with $\theta \equiv \arctan(c_2/c_3)$ .  (Signs are chosen to assure that $a$, $b$, and $b'$ are positive.)  

From the form of Eq.~(\ref{sprime}) it is clear that as long as either $c_2$ or $c_3$ is the smallest of the three output probabilities, there is a non-negative enhancement provided by entanglement.  The symmetry of this expression under interchange of $c_2$ and $c_3$ and the details of its derivation indicate  that it provides the optimal solution when either $c_2$ or $c_3$ is the smallest of the three probabilities.  If $c_1$ is smallest, this is not the case.  Alice should in this circumstance input $(\alpha,q)$ into channel ${\cal N}'$ (rather than $(q,\alpha)$.
  This switches the role of the first and second outputs, and the optimal communication enhancement is given by replacement of $c_2$ by $c_1$ in Eq.~{\ref{sprime}):
\begin{equation} {S'}_{max} = 1 + \frac{1}2 (\sqrt{c_1^2+c_3^2} - c_1 - c_3)  . \end{equation}

If any of the three probabilities $c_i$ are zero, a strategy can be chosen to have perfect communication possible using the classical channel alone.  In all other cases, entanglement provides enhancement of the classical channel's one-shot, successful-communication probability.

\section{Summary}

We have presented a formalism for evaluation of the degree of enhancement  that entanglement can provide for communication accuracy of information sent through a noisy classical channel, where both the channel used and the protocol employed are based on those presented in \cite{prevedel}.    We have utilized this formalism to derive the maximum enhancement of the classical communication channel that can be realized when the entanglement shared by sender and receiver resides in a two-qudit system, establishing also that qudits do not yield enhancement beyond that achievable using qubits.  A key simplification leading to the identification of the optimal protocol results from Eq.~(\ref{UnitaryFreedom}) of the Appendix,  employed with appropriate choices of unitary operators $U$ and $U^*$.   Our formalism and methods  should be useful in examining related entanglement-enhanced classical channel protocols.  

\section{Appendix \label{appx}}

\subsection{Simplifying the Objective Function}\label{SOF} We suppose that the initial state of  Alice and Bob's system of two fully entangled qudits is
\begin{equation}\label{FEQD}
\ket{\Phi} = \frac{1}{\sqrt{d}}\sum_{j=0}^{d-1} \ket{j}_A\ket{j}_B,
\end{equation}
where $d \ge 2$.  

\begin{lemma} Let $U$ be a unitary operator on ${\mathbb C}^d$ whose matrix representation with respect to Alice's basis $(\ket{j}_A)_{j=0}^{d-1}$ of ${\mathbb C}^d$ has entries $u_{ij}$, $0\le i,j\le d-1$.  Let $U^*$ be the unitary operator on ${\mathbb C}^d$ whose matrix with respect to Bob's basis $(\ket{j}_B)_{j=0}^{d-1}$ of ${\mathbb C}^d$ has entries $u_{i,j}^*$, $0\le i,j\le d-1$.  Then
$$
(U\otimes U^*)\ket{\Phi} = \ket{\Phi};
$$
that is, $\ket{\Phi}$ is an eigenvector of $U\otimes U^*$ with corresponding eigenvalue $1$.
\end{lemma}

\begin{proof}  For any matrix $M$, let $M^{(j)}$ denote the $j$-th row of $M$.   We have
\begin{eqnarray*}
(U\otimes U^*)\ket{\Phi} &=& \frac{1}{\sqrt{d}} \sum_{j=0}^{d-1}\sum_{i,k=0}^{d-1} u_{ij}u_{kj}^*\ket{i}_A\ket{k}_B\\
&=& \frac{1}{\sqrt{d}}\sum_{i,k=0}^{d-1} \braket{U^{(k)}}{U^{(i)}}\ket{i}_A\ket{k}_B\\
& =&  \frac{1}{\sqrt{d}}\sum_{i=0}^{d-1} \ket{i}_A\ket{i}_B.
\end{eqnarray*}
\end{proof}

 We assume that $\ket{\Phi}$ is given by (\ref{FEQD}) and that Alice measures in the directions  $\ket{\psi_0}$ and $\ket{\psi_1}$  while Bob measures in the directions  $\ket{\eta_0}$ and $\ket{\eta_1}$.  We extend $\{\ket{\psi_1}, \ket{\psi_0}\}$ to an orthonormal basis of ${\mathbb C}^d$ and choose $U$ to be a unitary operator on ${\mathbb C}^d$ that maps the first element in this basis, $\ket{\psi_1}$, to $\ket{0}_A$ and the second element in the basis to $\ket{1}_A$.  Then   $U\ket{\psi_1} = \ket{0}_A$ and $U\ket{\psi_0}$ is a linear combination of $\ket{0}_A$ and $\ket{1}_A$.  Define $U^*$ as in the statement of the preceding Lemma.  Because $(U\otimes U^*)\ket{\Phi} = \ket{\Phi}$, it is easy to see that
\begin{equation}\label{UnitaryFreedom}
F(\ket{\Phi}, \ket{\psi_0}, \ket{\psi_1}, \ket{\eta_0}, \ket{\eta_1}) = F(\ket{\Phi}, U\ket{\psi_0}, U\ket{\psi_1}, U^*\ket{\eta_0}, U^*\ket{\eta_1}),
\end{equation}
where $F$ is our objective function.  Thus, 
\begin{equation}\label{ROK}
\begin{split}
\max \{F(\ket{\Phi}, \ket{\psi}, \ket{0}_A, \ket{\eta}, \ket{\eta'}): \ket{\psi}, \ket{\eta}, \ket{\eta'}\ \text{are state vectors in}\  {\mathbb C}^d\}  = \rule{.75in}{0in}  \\ \max\{F(\ket{\Phi}, \ket{\psi_0}, \ket{\psi_1}, \ket{\eta_0}, \ket{\eta_1}): \ket{\psi_0}, \ket{\psi_1}, \ket{\eta_0},\ket{\eta_1}\ \text{are state vectors in}\ {\mathbb C}^d\};
\end{split}
\end{equation}
in fact, we may assume that $\ket{\psi}$ is a linear combination of $\ket{0}_A$ and $\ket{1}_A$.  

\subsection{Enhancement When $\ket{\Phi}$ is a Fully Entangled Two-Qudit State }\label{FETQDS}

We now prove that if the system that Alice and Bob share comprises two fully entangled qudits, then the maximum enhancement value achievable is 
\[ \frac2d F_{\max} , \]
where $F_{\max} = (4 \cos^2(\pi/8)-3)/6$ is the maximum enhancement possible for qubits.  Thus we assume that $\ket{\Phi}$ is in the fully entangled two-qudit state  (\ref{FEQD}).  Using the inequality (\ref{ROK}) with $\ket{\psi}$ restricted to be a linear combination of $\ket{0}_A$ and $\ket{1}_A$ , we see that we may assume in Eq.~(\ref{PhiQD}), that $\ket{p_0} = \ket{0}$ and $\ket{p_1}$ is a linear combination of $\ket{0}$ and $\ket{1}$.   Because $\ket{\Phi}$ is in the fully entangeld state (\ref{FEQD}), Eq.~(\ref{PhiQD}), becomes
\begin{equation}\label{FER}
\frac{1}{\sqrt{d}}\sum_{j=0}^{d-1} \ket{j}_A\ket{j}_B  =  \sum_{i,j=0}^{d-1} m_{ij}\ket{p_i}\ket{n_j}.
\end{equation}
Because $\ket{p_i}$, for $i\ge 2$, is orthogonal to linear combinations of $\ket{p_0}$ and $\ket{p_1}$, it must be orthogonal to both $\ket{0}$ and $\ket{1}$.  It follows from Eq.\ (\ref{FER}) that 
\begin{equation}\label{GT}
\frac{1}{\sqrt{d}}(\ket{00} + \ket{11}) = \sum_{i=0}^1\sum_{j=0}^{d-1} m_{ij}\ket{p_i}\ket{n_j}.
\end{equation}
Recalling that  $\ket{\tilde{\Phi}}$ from Eq.~(\ref{PhiTQD}) is given by
$$
\ket{\tilde{\Phi}} = \sum_{i,j=0}^1 m_{ij}\ket{p_i}\ket{n_j}.
$$
we see, thanks to (\ref{GT}), that
$$
\|\ket{\tilde{\Phi}}\|^2 \le  \left\|\frac{1}{\sqrt{d}}(\ket{00} + \ket{11})\right\|^2 = \frac{2}{d}.
$$
Because, by Eq~(\ref{qudit}),  the probability of success starting with two qudits is at best $5/6 + \|\tilde{\Phi}\|^2F_{\max}$, the preceding equation shows the probability of success starting with two {\it fully entangled qudits} is at best   $5/6 + \frac{2}{d}F_{\max}$ (a bound that is achievable by choosing $\ket{\phi} = \cos(\pi/4)\ket{0} + \sin(\pi/4)\ket{1}$, $\ket{\psi'} = \ket{0}$, $\ket{\eta} = \cos(\pi/8) \ket{0} + \sin(\pi/8)\ket{1}$, and $\ket{\eta'} = \cos(3\pi/8)\ket{0} + +\sin(3\pi/8)\ket{1}$).

\end{document}